\documentclass[conference]{IEEEtran}
\ifCLASSINFOpdf
   \usepackage[pdftex]{graphicx}
\else
   \usepackage[dvips]{graphicx}
\fi
\usepackage[cmex10]{amsmath}
\usepackage{amssymb}
\usepackage[font=footnotesize]{subfig}
\usepackage{pst-all}
\usepackage{pstricks}
\begin{document}
%
\title{Optimal Power Allocation Policy over Two Identical Gilbert-Elliott Channels}

\author{\IEEEauthorblockN{Wei Jiang}
\IEEEauthorblockA{School of Information Security\\
 Engineering\\
Shanghai Jiao Tong University, China\\
Email: kerstin@sjtu.edu.cn}
\and
\IEEEauthorblockN{Junhua Tang}
\IEEEauthorblockA{School of Information Security\\
 Engineering\\
Shanghai Jiao Tong University, China\\
Email: junhuatang@sjtu.edu.cn}
\and
\IEEEauthorblockN{Bhaskar Krishnamachari}
\IEEEauthorblockA{Ming Hsieh Department \\
 of Electrical Engineering\\
Viterbi School of Engineering \\
University of Southern California\\
Email: bkrishna@usc.edu}
}


\maketitle

\newtheorem{lemma}{\textbf{Lemma}}
\newtheorem{theorem}{\textbf{Theorem}}
\newtheorem{definition}{\textbf{Definition}}
\newtheorem{observation}{\textbf{Observation}}

\begin{abstract}
We study the fundamental problem of optimal power allocation over two identical Gilbert-Elliott (Binary Markov) communication channels. Our goal is to maximize the expected discounted number of bits transmitted over an infinite time span by judiciously choosing one of the four actions for each time slot: 1) allocating power equally to both channels, 2) allocating all the power to channel 1, 3) allocating all the power to channel 2, and 4) allocating no power to any of the channels. As the channel state is unknown when power allocation decision is made, we model this problem as a partially observable Markov decision process(POMDP), and derive the optimal policy which gives the optimal action to take under different possible channel states. Two different structures of the optimal policy are derived analytically and verified by linear programming simulation. We also illustrate how to construct the optimal policy by the combination of threshold calculation and linear programming simulation once system parameters are known.
\end{abstract}


%
\IEEEpeerreviewmaketitle

\section{Introduction}
Adaptive power control is an important technique to select the transmission power of a wireless system according to channel condition to achieve better network performance in terms of higher data rate or spectrum efficiency [1],[2]. There has been some recent work on power allocation over stochastic channels [3],[4],[5]; the problem of optimal power allocation across multiple dynamic stochastic channels is challenging and remains largely unsolved from a theoretical perspective

We consider a wireless system operating on two parallel transmission channels. The two channels are statistically identical and independent of each other. We model each channel as a slotted Gilbert-Elliott channel. That is, each channel is described by a two-state Markov chain, with a bad state ``0'' and a good state ``1'' [7]. Our objective is to allocate the limited power budget to the two channels dynamically so as to maximize the expected discounted number of bits transmitted over time. Since the channel state is unknown when the decision is made, this problem is more challenging than it looks like.

Recently, several works have explored different sequential decision-making problems involving Gilbert-Elliott channels. In [8],[9], the authors consider the problem of selecting one channel to sense/access among several identical channels, formulate it as a restless multi-armed bandit problem, and show that a simple myopic policy is optimal whenever the channels are positively correlated over time. In [6], the authors study the problem of dynamically choosing one of three transmitting schemes for a single Gilbert-Elliott channel in an attempt to maximize the expected discounted number of bits transmitted. And in [10], the authors study the problem of choosing a transmitting strategy from two choices emphasizing the case when the channel transition probabilities are unknown. While similar in spirit to these two studies, our work addresses a more challenging setting involving two independent channels. In [6],[8],[9], only one channel is accessed in each time slot, while our formulation of power allocation is possible to use both channels simultaneously. In [17], a similar power allocation problem is studied. Our work in this paper has the following differences compared with the work in [17]: four power allocation actions are considered instead of 3; penalty is introduced when power is allocated to a channel in bad condition. With the introduction of one more action (using none of the two channels) and transmission penalty, the problem becomes more interesting yet more difficult to analyze.

In this paper, we formulate our power allocation problem as a partially observable Markov decision process(POMDP). We then convert it to a continuous state Markove Decision Process (MDP) problem and derive the structure of the optimal policy. Our main contributions are:(1)we formulate the problem using the MDP theory and theoretically prove the structure of the optimal policy, (2) we verify our analysis through simulation based on linear programming, (3) we demonstrate how to numerically obtain the structure of this optimal policy when system parameters are known.

The results in this paper advance the fundamental understanding of optimal power allocation over multiple dynamic stochastic channels from a theoretical perspective.

\section{Problem Formulation}
\subsection{Channel model and assumptions}
We consider a wireless communication system operating on two parallel channels. Each channel is described by a slotted Gilbert-Elliott model which is a one dimensional two-state Markov chain $G_{i,t}(i \in \{1,2\} , t \in \{1,2,...,\infty\})$:
a good state denoted by 1 and a bad state denoted by 0 ($i$ is channel index and $t$ is time slot). The state transition probabilities are: $P_r[G_{i,t}=1 | G_{i,t-1}=1] = \lambda_1$ and $P_r[G_{i,t}=1 | G_{i,t-1}=0] = \lambda_0, i \in \{1,2\}$. We assume the two channels are identical and independent of each other. Meanwhile channel state transition occurs only at the beginning of each time slot. We also assume that $\lambda_0 < \lambda_1$, which is a positive correlation assumption commonly used in the literature.

The system has a total power $P$. At the beginning of each slot, the system allocates power $P_1(t)$ to channel 1 and power $P_2(t)$ to channel 2, where $P_1(t)+P_2(t) = P$.
We assume channel state is unknown at the beginning of each time slot, thus the system needs to decide the power allocation for the two channels without knowing the channel states. If a channel is used in slot $t$, its channel state during that slot is revealed at the end of time slot $t$ through channel feedback. But if a channel is not used, its state during the elapsed time slot remains unknown.

\subsection{Power allocation strategies}
To simplify the power allocation problem, we define three power levels the system may allocate to a channel: $0,P/2,P$. If a channel in good state is allocated power $P/2$, it can transmit $R_l$ bits of data during that slot. If a channel in good state is allocated power $P$, it can transmit $R_h$ bits of data successfully. We assume $R_l < R_h < 2R_l$. At the same time, if a channel in bad state is allocated power $P/2$, it suffers $C_l$ bits of data loss. If a channel in bad state is allocated power $P$, it suffers $C_h$ bits of data loss. We assume $C_l < C_h < 2C_l$ and $R_h > C_h, R_l > C_l$.

At the beginning of each time slot, the system chooses one the following four actions: balanced, betting on channel 1, betting on channel 2 and conservative.

\emph{Balanced} (denoted by $B_b$): the system allocates power evenly on both channels, that is, $P_1(t)=P_2(t)=P/2$ for time slot $t$. This action is chosen when the system believes both channels are in good state and it is most beneficial to use both of the channels.

\emph{Betting on channel 1} (denoted by $B_1$): the system decides to ``gamble'' by allocating all the power to channel 1, that is, $P_1(t)=P,P_2(t)=0$. This occurs when the system believes that channel 1 will be in good state and channel 2 will be in bad state.

\emph{Betting on channel 2} (denoted by $B_2$): contrary to $B_1$, the system allocates all the power to channel 2, that is, $P_1(t)=0,P_2(t)=P$.

\emph{Conservative} (denoted by $B_r$): the system decides to ``play safe'' by using none of the two channels, that is, $P_1(t)=P_2(t)=0$. This action is taken when the system believes both channels will be in bad state and using any of the channels will cause data loss.

Note that in actions $B_1$, $B_2$ and $B_r$, if a channel is not used, the system will not know its state in the elapsed slot.

\subsection{Formulation of the Partially Observable Markov Decision problem}
At the beginning of each time slot, the system needs to judiciously choose one of the four power allocation actions to maximize the total discounted number of data bits transmitted over an infinite time span. Since the channel state is not observable when the choice is made, this power allocation problem is a Partially Observable Markov Decision Problem (POMDP). In [11], it shows that a sufficient statistic for determining the optimal action is the conditional probability that the channel is in good state at the beginning of the current slot given the past history, henceforth this conditional probability is called belief. We denote the belief by a two dimensional vector $\mathbf{x}_t=(x_{1,t},x_{2,t})$, where $x_{i,t}=P_r[G_{i,t}=1 | \hbar_t], i \in \{1,2\}$, $\hbar_t$ is all the history of actions and state observations prior to the beginning of current slot. Using the belief as decision variable, the POMDP problem is converted into an MDP problem with an uncountable state space $([0,1],[0,1])$ [8].

Define a policy $\pi$ as a rule that determines the action to take under different situations, that is, a mapping from the belief space to action space. Let $V^{\pi}(\mathbf{p})$ denote the expected discounted reward with initial belief $\mathbf{p}=(p_1,p_2)$, that is, $x_{1,0}=P_r[G_{1,0}=1 | \hbar_0]=p_1, x_{2,0}=P_r[G_{2,0}=1 | \hbar_0]=p_2$, with $\pi$ denoting the policy followed. With discount factor $\beta \in [0,1]$, the expected discounted reward is expressed as
\begin{equation}
V^{\pi}(\mathbf{p}) = E^{\pi}[\sum^{\infty}_{t=0} \beta^{t} g_{a_t}(\mathbf{x}_t) | \mathbf{x}_0 = \mathbf{p}],
\end{equation}
where $E^\pi$ denotes the expectation given policy $\pi$, $t$ is the time slot index, $a_t \in \{B_1,B_2,B_b,B_r\}$ represents the action taken at time t. The term $g_{a_t}(\mathbf{x}_t)$ denotes the expected immediate reward when the belief is $\mathbf{x}_t$ and action $a_t$ is chosen:
\begin{equation}
g_{a_t}(\mathbf{x}_t) = \left\{ \begin{array}{ll}
x_{1,t}(R_h + C_h) - C_h & \textrm{if $a_t=B_1$} \\
x_{2,t}(R_h + C_h) - C_h & \textrm{if $a_t=B_2$} \\
(x_{1,t}+x_{2,t})(R_l + C_l) - 2C_l & \textrm{if $a_t = B_b$} \\
0 & \textrm{if $a_t=B_r$}
\end{array}. \right.
\end{equation}
Now we define the value function $V(\mathbf{p})$ as
\begin{equation}
V(\mathbf{p}) = \max_{\pi}V^{\pi}(\mathbf{p}) \quad \forall \quad \mathbf{p} \in ([0,1],[0,1]).
\end{equation}
A policy is stationary if it is a function mapping the state space $([0,1],[0,1])$ into action space $\{B_1,B_2,B_b,B_r\}$. Ross proved in [12](Th.6.3) that there exists a stationary policy $\pi^{*}$ such that $V(\mathbf{p})=V^{\pi^*}(\mathbf{p})$, and the value function $V(\mathbf{p})$ satisfies the Bellman equation
\begin{equation}
V(\mathbf{p}) = \max_{a \in \{B_1,B_2,B_b,B_r\}} \{V_a(\mathbf{p})\},
\end{equation}
where $V_a(\mathbf{p})$ denotes the value acquired when the belief is $\mathbf{p}$ and action $a$ is taken. $V_a(\mathbf{p})$ is given by
\begin{equation}
V_a(\mathbf{p}) = g_a(\mathbf{p}) + \beta E^y[V(y) | \mathbf{x}_0=\mathbf{p},a_0=a],
\end{equation}
where y denotes the next belief after action $a$ is taken when the initial belief is $\mathbf{p}$. $V_a(\mathbf{p})$ for the four actions is derived as follows.

a) \emph{Balanced($B_b$)}: If this action is taken with initial belief $\mathbf{p}=(p_1,p_2)$, the immediate reward is $p_1R_l+p_2R_l$ and the immediate loss is $(1-p_1)C_l+(1-p_2)C_l$. Since both channels are used, their states during the current slot are revealed at the end of current time slot. Therefore with probability $p_1$ channel 1 will be in good state hence the belief of channel 1 at the beginning of the next slot will be $\lambda_1$. Likewise, with probability $1-p_1$ channel 1 will be in bad state thus the belief in the next slot will be $\lambda_0$. Since both channels are identical, channel 2 has similar belief update. Consequently, the value function when action $B_b$ is taken can be expressed as
\begin{equation}
\label{eqn:VBb}
\begin{array}{ll}
  & V_{B_b}(\mathbf{p}) \\
= & (p_1+p_2)(R_l+C_l)-2C_l \\
+ & \beta[(1-p_1)(1-p_2)V(\lambda_0,\lambda_0) + p_1p_2V(\lambda_1,\lambda_1)\\
+ & p_1(1-p_2)V(\lambda_1,\lambda_0) + (1-p_1)p_2V(\lambda_0,\lambda_1)]
\end{array}.
\end{equation}

b) \emph{Betting on channel 1($B_1$)}: If this action is taken with initial belief $\mathbf{p}=(p_1,p_2)$, the immediate reward is $p_1R_h$, and the immediate loss is $(1-p_1)C_h$. Since channel 2 is not used, its channel state in the current slot remains unknown. Therefore the belief of channel 2 in the next time slot is calculated as
\begin{equation}
T(p_2) = (1-p_2)\lambda_0 + p_2\lambda_1 = \alpha p_2 + \lambda_0,
\end{equation}
where $\alpha=\lambda_1-\lambda_0$. Consequently, the value function when action $B_1$ is taken can be expressed as
\begin{equation}
\label{eqn:VB1}
\begin{array}{ll}
  & V_{B_1}(\mathbf{p}) \\
= & (R_h+C_h)p_1 - C_h  \\
+ & \beta[p_1V(\lambda_1,T(p_2)) + (1-p_1)V(\lambda_0,T(p_2))]
\end{array}.
\end{equation}

c) \emph{Betting on channel 2($B_2$)}: Similar to action $B_1$, the value function when actin $B_2$ is taken can be expressed as
\begin{equation}
\label{eqn:VB2}
\begin{array}{ll}
  & V_{B_2}(\mathbf{p}) \\
= & (R_h+C_h)p_2 - C_h  \\
+ & \beta[p_2V(T(p_1),\lambda_1) + (1-p_2)V(T(p_1),\lambda_0)]
\end{array},
\end{equation}
where
\begin{equation}
T(p_1) = (1-p_1)\lambda_0 + p_1\lambda_1 = \alpha p_1 + \lambda_0.
\end{equation}

d) \emph{Conservative($B_r$)}: If this action is taken, both immediate reward and loss are 0. Since none of the channel is used, their belief at the beginning of the next slot is given by
\begin{equation}
T(p_i) = (1-p_i)\lambda_0 + p_i\lambda_1 = \alpha p_i + \lambda_0, \quad i \in \{1,2\}.
\end{equation}
Consequently, the value function when action $B_r$ is taken can be expressed as
\begin{equation}
\label{eqn:VBr}
V_{B_r}(\mathbf{p})= \beta V(T(p_1),T(p_2)).
\end{equation}

Finally, the Bellman equation for our power allocation problem reads as
\begin{equation}
V(\mathbf{p}) = \max \{V_{B_b}, V_{B_1}, V_{B_2}, V_{B_r}\}.
\end{equation}

\section{Structure of the Optimal Policy}
From the discussion in the previous section, we understand that an optimal policy exists for our power allocation problem. In this section, we try to derive the optimal policy by first looking at the features of its structure.

\subsection{Properties of value function}
\begin{lemma}
$V_a(\mathbf{p}), a \in \{B_b,B_1,B_2,B_r\}$ is affine in both $p_1$ and $p_2$ and the following equalities hold:
\begin{displaymath}
\begin{split}
V_a(cp+(1-c)p', p_2) = cV_a(p,p_2) + (1-c)V_a(p',p_2) \\
V_a(p_1, cp+(1-c)p') = cV_a(p_1,p) + (1-c)V_a(p_1,p')
\end{split},
\end{displaymath}
where $0 \le c \le 1$ is a constant, and $f(x)$ is said to be affine with respect to $x$ if $f(x)=ax+c$ with constant $a$ and $c$.
\end{lemma}
\begin{proof}
It is clear from (\ref{eqn:VBb}) that $V_{B_b}$ is affine in $p_1$ and $p_2$. Also it is obvious that $V_{B_1}$ is affine in $p_1$ and $V_{B_2}$ is affine in $p_2$ from (\ref{eqn:VB1}) and (\ref{eqn:VB2}), respectively. Next, we will prove that $V_{B_1}$ is affine in $p_2$.

Let's look at the right side of equation \ref{eqn:VB1}. The first and second terms are not related to $p_2$ so this part is affine in $p_2$. For the third term, the main part $V(c,T(p_2))$ ($c \in \{\lambda_0,\lambda_1\}$) takes one of the following four forms: $V_{B_b}(c,T(p_2)), V_{B_2}(c,T(p_2)), V_{B_1}(c,T(p_2))$ or $V_{B_r}(c,T(p_2))$. The first form is affine in $p_2$ because $V_{B_b}(c,T(p_2))$ is affine in $T(p_2)$ and $T(p_2)=\alpha p_2 + \lambda_0$ is affine in $p_2$. Similarly the second form $V_{B_2}(c,T(p_2))$ is affine in $T(p_2)$ thus also affine in $p_2$. For the latter two forms $V_{B_1}(c,T(p_2))$ and $V_{B_r}(c,T(p_2))$, they can be written as:
\begin{equation}
\label{eqn:VB1lam0Tp2}
\begin{array}{ll}
 & V_{B_1}(c,T(p_2)) \\
=& c(R_h+C_h)-C_h \\
+& \beta c V(\lambda_1,T^2(p_2)) + \beta (1-c)V(\lambda_0,T^2(p_2))
\end{array},
\end{equation}
or
\begin{equation}
\label{eqn:VBrlam0p2}
 V_{B_r}(c,T(p_2)) = V_{B_r}(T(c),T^2(p_2)),
\end{equation}
where $T^n(p)=T^{(n-1)}(T(p))=\frac{\lambda_0}{1-\alpha}(1-\alpha^n)+\alpha^np$. 
Since $T^n(p_2)$ is affine in $p_2$, 
(\ref{eqn:VB1lam0Tp2}) is affine in $p_2$ as soon as $V(\lambda_1,T^2(p_2))$ takes the form $V_{B_b}(\lambda_1,T^n(p_2))$ or $V_{B_2}(\lambda_1,T^n(p_2))$, and $V(\lambda_0,T^2(p_2))$ takes the form $V_{B_b}(\lambda_0,T^n(p_2))$ or $V_{B_2}(\lambda_0,T^n(p_2))$, $n = 2,3,\cdots$, which is affine in $p_2$. If $V(\lambda_1,T^2(p_2))$ keeps taking the form $V_{B_1}(\lambda_1,T^n(p_2))$ or $V_{B_r}(\lambda_1,T^n(p_2))$ till $n$ goes to infinity, it will eventually become $V_{B_1}(\lambda_1,\frac{\lambda_0}{1-\alpha})$ or $V_{B_r}(\lambda_1,\frac{\lambda_0}{1-\alpha})$ because $T^n(p_2) \to \frac{\lambda_0}{1-\alpha}$ when $n \to \infty$, which is a special case of affine in $p_2$. The same is true for the term $V(\lambda_0,T^2(p_2))$. With this we show that (\ref{eqn:VB1lam0Tp2}) is affine in $p_2$. Similarly we can prove that (\ref{eqn:VBrlam0p2}) is affine in $p_2$ thus (\ref{eqn:VB1}) is affine in $p_2$.

Using the same technique we can prove that $V_{B_2}(p_1,p_2)$ is affine in $p_1$ and $V_{B_r}$ is affine in both $p_1$ and $p_2$. With this we show that $V_a(\mathbf{p}), a \in \{B_b,B_1,B_2,B_r\}$ is affine in both $p_1$ and $p_2$, and the equalities in Lemma 1 immediately follow. This concludes the proof.
\end{proof}

\begin{lemma}
$V(\mathbf{p})$ is convex in $p_1$ and $p_2$, and the following inequalities hold:
\begin{displaymath}
\begin{array}{lll}
 V(cp+(1-c)p', p_2) &\leq& cV(p,p_2) + (1-c)V(p',p_2) \\
 V(p_1, cp+(1-c)p') &\leq& cV(p_1,p) + (1-c)V(p_1,p')
\end{array}.
\end{displaymath}
\end{lemma}
\begin{proof}
The convexity property of the value function of any general POMDP is proved in [11] and we will use that result directly in this paper.
\end{proof}

\begin{lemma}
$V(p_1,p_2)=V(p_2,p_1)$, that is, $V(\mathbf{p})$ is symmetric with respect to the line $p_1=p_2$ in the belief space.
\end{lemma}
\begin{proof}
Let $V^n(p_1,p_2)$ denote the expected reward when the decision horizon spans only n time slots. When $n=1$,
\begin{equation}
\begin{array}{ll}
 & V^1(p_1,p_2) \\
=& \max \{(p_1+p_2)(R_l+C_l)-2C_l, 0, \\
 & p_1(R_h+C_h)-C_h, p_2(R_h+C_h)-C_h \}
\end{array}.
\end{equation}
\begin{equation}
\begin{array}{ll}
 & V^1(p_2,p_1) \\
=& \max \{(p_1+p_2)(R_l+C_l)-2C_l, 0, \\
 &  p_2(R_h+C_h)-C_h, p_1(R_h+C_h)-C_h \}
\end{array}.
\end{equation}
Obviously we have $V^1(p_1,p_2) = V^1(p_2,p_1)$. Next we assume $V^k(p_1,p_2)=V^k(p_2,p_1), k \ge 1$, we now show that $V^{k+1}(p_1,p_2)=V^{k+1}(p_2,p_1)$. Since
\begin{equation}
\begin{array}{ll}
 & V^{k+1}_{B_b}(p_1,p_2) \\
=& (p_1+p_2)(R_l+C_l)-2C_l + \beta [ p_1p_2V^{k}(\lambda_1,\lambda_1) \\
+& p_1(1-p_2)V^{k}(\lambda_1,\lambda_0) + (1-p_1)p_2V^{k}(\lambda_0,\lambda_1) \\
+& (1-p_1)(1-p_2)V^{k}(\lambda_0, \lambda_0)
\end{array}.
\end{equation}
\begin{equation}
\begin{array}{ll}
 & V^{k+1}_{B_1}(p_1,p_2) \\
=& p_1(R_h+C_h)-C_h + \\
 & \beta [(1-p_1)V^{k}(\lambda_0,T(p_2)) + p_1V^{k}(\lambda_1,T(p_2))] \\
\end{array}.
\end{equation}
\begin{equation}
\begin{array}{ll}
 & V^{k+1}_{B_2}(p_1,p_2) \\
=& p_2(R_h+C_h)-C_h + \\
 & \beta [(1-p_2)V^{k}(T(p_1),\lambda_0) + p_2V^{k}(T(p_1),\lambda_1)] \\
\end{array}.
\end{equation}
\begin{equation}
\begin{array}{l}
 V^{k+1}_{B_r}(p_1,p_2) = \beta V^{k}(T(p_1),T(p_2))
\end{array}.
\end{equation}
Using the assumption that $V^k(p_1,p_2)=V^k(p_2,p_1)$, we have,
\begin{equation}
\begin{array}{ll}
 & V^{k+1}_{B_1}(p_1,p_2) \\
=& p_1(R_h+C_h)-C_h + \\
 & \beta [(1-p_1)V^{k}(T(p_2),\lambda_0) + p_1V^{k}(T(p_2),\lambda_1)] \\
=& V^{k+1}_{B_2}(p_2,p_1)
\end{array}.
\end{equation}
Similarly, we have $V^{k+1}_{B_2}(p_1,p_2)=V^{k+1}_{B_1}(p_2,p_1)$, $V^{k+1}_{B_b}(p_1,p_2)=V^{k+1}_{B_b}(p_2,p_1)$ and $V^{k+1}_{B_r}(p_1,p_2)=V^{k+1}_{B_r}(p_2,p_1)$. Therefore,
\begin{equation}
\begin{array}{ll}
 & V^{(k+1)}(p_1,p_2) \\
=& \max \{V^{k+1}_{B_b}(p_1,p_2), V^{k+1}_{B_1}(p_1,p_2), \\
 & V^{k+1}_{B_2}(p_1,p_2), V^{k+1}_{B_r}(p_1,p_2),\} \\
=& \max \{V^{k+1}_{B_b}(p_2,p_1), V^{k+1}_{B_2}(p_2,p_1), \\
 & V^{k+1}_{B_1}(p_2,p_1), V^{k+1}_{B_r}(p_2,p_1),\} \\
=& V^{(k+1)}(p_2,p_1)
\end{array}.
\end{equation}
Hence we have $V(p_1,p_2)=V(p_2,p_1)$ for all $(p_1,p_2)$ in the belief space.
\end{proof}

\subsection{Properties of the decision regions of policy $\pi^{*}$}
We use $\Phi_a$ to denote the decision region of action $a$. That is, $\Phi_a$ is the set of beliefs under which it is optimal to take action $a$:
\begin{equation}
\begin{array}{r}
\Phi_a = \{(p_1,p_2) \in ([0,1],[0,1]) |  V(p_1,p_2) = V_a(p_1,p_2) \} \\
 a \in \{B_b,B_1,B_2,B_r\}
\end{array}.
\end{equation}
\begin{definition}
$\Phi_a$ is said to be contiguous along $p_1$ dimension if given $(x_1,p_2),(x_2,p_2) \in \Phi_a$, then $\forall x \in [x_1,x_2]$, we have $(x,p_2) \in \Phi_a$. Similarly, we say $\Phi_a$ is contiguous along $p_2$ dimension if given $(p_1,y_1),(p_1,y_2) \in \Phi_a$, then $\forall y \in [y_1,y_2]$, we have $(p_1,y) \in \Phi_a$.
\end{definition}
\begin{theorem}
$\Phi_a$ is contiguous in both $p_1$ and $p_2$, where $a \in \{B_b,B_1,B_2,B_r\}$.
\end{theorem}
\begin{proof}
We will prove $\Phi_{B_1}$ as an example, and the results for other actions can be proved in a similar manner. First we prove that $\Phi_{B_1}$ is contiguous in $p_1$. Let $(x_1,p_2),(x_2,p_2) \in \Phi_{B_1}$, next we show that $((cx_1+(1-c)x_2),p_2)$ is also in region $\Phi_{B_1}$, where $0 \le c \le 1$.
\begin{equation}
\begin{array}{ll}
   & V((cx_1+(1-c)x_2),p_2) \\
\le& cV(x_1,p_2) + (1-c)V(x_2,p_2) \\
  =& cV_{B_1}(x_1,p_2) + (1-c)V_{B_1}(x_2,p_2) \\
  =& V_{B_1}((cx_1+(1-c)x_2),p_2) \\
\le& V((cx_1+(1-c)x_2),p_2)
\end{array},
\end{equation}
where the first inequality comes from the convexity in lemma 2; the first equality follows from the fact that $(x_1,p_2),(x_2,p_2) \in \Phi_{B_1}$; the second equality follows from the affine linearity of $V_{B_1}(\mathbf{p})$ in $p_1$; the last inequality follows from the definition of $V(\mathbf{p})$. We have $V((cx_1+(1-c)x_2),p_2)=V_{B_1}((cx_1+(1-c)x_2),p_2)$, that is, $((cx_1+(1-c)x_2),p_2)$ is also in the region $\Phi_{B_1}$. Therefore, $\Phi_{B_1}$ is contiguous in $p_1$. Similarly $\Phi_{B_1}$ is contiguous in $p_2$.
\end{proof}

\begin{theorem}
$\Phi_{B_b}$ and $\Phi_{B_r}$ are self-symmetric with respect to the line $p_1=p_2$, that is, if $(p_1,p_2) \in \Phi_a, a \in \{B_b,B_r\}$ then $(p_2,p_1) \in \Phi_a$. $\Phi_{B_1}$ and $\Phi_{B_2}$ are mirrors with respect to the line $p_1=p_2$, that is,
if $(p_1,p_2) \in \Phi_{B_1}$ then $(p_2,p_1) \in \Phi_{B_2}$.
\end{theorem}
\begin{proof}
If $(p_1,p_2) \in \Phi_{B_r}$, then we have
\begin{equation}
V(p_1,p_2) = V_{B_r}(p_1,p_2).
\end{equation}
Using lemma 3, we have
\begin{equation}
\begin{array}{ll}
 & V_{B_r}(p_2,p_1) \\
=& \beta V(T(p_2),T(p_1)) = \beta V(T(p_1),T(p_2)) \\
=& V_{B_r}(p_1,p_2) = V(p_1,p_2) = V(p_2,p_1)
\end{array},
\end{equation}
hence $(p_2,p_1)$ also belongs to $\Phi_{B_r}$. Similarly, we can show that if $(p_1,p_2) \in \Phi_{B_b}$, then $(p_2,p_1)$ also belongs to $\Phi_{B_b}$.

If $(p_1,p_2) \in \Phi_{B_1}$, then we have
\begin{equation}
\begin{array}{ll}
 & V(p_1,p_2) = V_{B_1}(p_1,p_2) \\
=& p_1(R_h+C_h)-C_h \\
+& \beta [p_1V(\lambda_1,T(p_2)) + (1-p_1)V(\lambda_0,T(p_2))]
\end{array}.
\end{equation}
Using lemma 3, we have
\begin{equation}
\begin{array}{ll}
 & V_{B_2}(p_2,p_1) \\
=& p_1(R_h+C_h)-C_h \\
+& \beta [p_1V(T(p_2),\lambda_1) + (1-p_1)V(T(p_2),\lambda_0)] \\
=& p_1(R_h+C_h)-C_h \\
+& \beta [p_1V(\lambda_1,T(p_2)) + (1-p_1)V(\lambda_0,T(p_2))] \\
=& V_{B_1}(p_1,p_2) = V(p_1,p_2) = V(p_2,p_1)
\end{array},
\end{equation}
hence $(p_2,p_1)$ belongs to $\Phi_{B_2}$ which concludes the proof.
\end{proof}

\subsection{Structure of the optimal policy}
Based on the properties discussed above, we are now ready to derive the structure of the optimal policy.

From the belief update in (\ref{eqn:VBb}), (\ref{eqn:VB1}),(\ref{eqn:VB2}) and (\ref{eqn:VBr}), it is clear that the belief state of a channel is updated to one of the following three values after any action: $\lambda_0$, $\lambda_1$, or $T(p)$, where $p$ is the current belief of a channel. For all $0 \leq p \leq 1$, $\lambda_0 \leq T(p)= \lambda_0 + (\lambda_1-\lambda_0)p \leq \lambda_1$. Since $0 \leq \lambda_0, \lambda_1 \leq 1$, the belief space is the rectangle area determined by four vertices at $(0,0), (0,1), (1,1)$ and $(1,0)$.

First we consider the four vertices and it is easy to obtain the following results,
\begin{equation}
\left\{ \begin{array}{l}
V(0,0) = V_{B_r}(0,0) \\
V(0,1) = V_{B_2}(0,1) \\
V(1,0) = V_{B_1}(1,0) \\
V(1,1) = V_{B_b}(1,1)
\end{array} \right.
\Rightarrow
\left\{ \begin{array}{l}
(0,0) \in \Phi_{B_r} \\
(0,1) \in \Phi_{B_2} \\
(1,0) \in \Phi_{B_1} \\
(1,1) \in \Phi_{B_b}
\end{array}. \right.
\end{equation}
Next we consider the four edges. On the edge $p_1=0$, the partial value functions are
\begin{equation}
\left\{ \begin{array}{lll}
V_{B_1}(0,p_2) & = & -C_h + \beta V(\lambda_0,T(p_2)) \\
V_{B_2}(0,p_2) & = &p_2(R_h+C_h)-C_h + \\
               &  &\beta [(1-p_2)V(\lambda_0,\lambda_0) + p_2V(\lambda_0,\lambda_1)] \\
V_{B_b}(0,p_2) & = &p_2(R_l+C_l)-2C_l + \\
               &  &\beta [(1-p_2)V(\lambda_0,\lambda_0) + p_2V(\lambda_0,\lambda_1)] \\
V_{B_r}(0,p_2) & = &\beta V(\lambda_0,T(p_2))
\end{array}. \right.
\end{equation}
Using our assumption $C_l<C_h<2C_l, R_h>C_h, R_l>C_l$, and convexity of value function $V(p_1,p_2)$, we have $V_{B_1}(0,p_2)<V_{B_b}(0,p_2)<V_{B_2}(0,p_2)$. With this we say that on the edge $p_1=0$, only two actions $B_2$ and $B_r$ are possible. Since $(0,0) \in \Phi_{B_r}$  and $(0,1) \in \Phi_{B_2}$, we know there exists a threshold $\rho$ such that $\forall p_2 \in [0,\rho)$, $(0,p_2) \in \Phi_{B_r}$ and $\forall p_2 \in [\rho,1], (0,p_2) \in \Phi_{B_2}$. To derive $\rho$ we define
\begin{equation}
\left\{ \begin{array}{l}
\delta_{1,\lambda_0}(p) = (1-p)V(\lambda_0,\lambda_0) + p V(\lambda_1,\lambda_0) - V(T(p),\lambda_0) \\
\delta_{1,\lambda_1}(p) = (1-p)V(\lambda_0,\lambda_1) + p V(\lambda_1,\lambda_1) - V(T(p),\lambda_1) \\
\delta_{2,\lambda_0}(p) = (1-p)V(\lambda_0,\lambda_0) + p V(\lambda_0,\lambda_1) - V(\lambda_0,T(p)) \\
\delta_{2,\lambda_1}(p) = (1-p)V(\lambda_1,\lambda_0) + p V(\lambda_1,\lambda_1) - V(\lambda_1,T(p))
\end{array}. \right.
\end{equation}
From the symmetric property of $V(p_1,p_2)$, we have $\delta_{1,\lambda_0}(p)=\delta_{2,\lambda_0}(p)=\delta_{\lambda_0}(p)$ and $\delta_{1,\lambda_1}(p)=\delta_{2,\lambda_1}(p)=\delta_{\lambda_1}(p)$. Using the fact that $V_{B_2}(0,\rho) = V_{B_r}(0,\rho)$, we have
\begin{equation}
\begin{array}{ll}
 & V_{B_2}(0,\rho)-V_{B_r}(0,\rho) \\
=& \rho(R_h+C_h)-C_h+\beta \delta_{\lambda_0}(\rho) = 0
\end{array}
\end{equation}
\begin{equation}
\rho = \frac{C_h-\beta\delta_{\lambda_0}(\rho)}{R_h+C_h}.
\end{equation}

Using the results in Theorem 2, we can easily derive similar structure on the other three edges. The structure of the optimal policy on the boundary of the belief space is shown in Fig.1
\begin{figure}[htbp]
\label{fig:boundary}
\centering
\includegraphics[width=0.3\textwidth,natwidth=381,natheight=362]{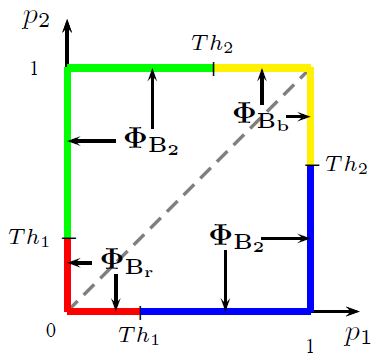}
\caption{Structure of the optimal policy on the boundary of belief space} 
\end{figure}
The thresholds $Th_1$ and $Th_2$ in Figure 1 are given by
\begin{equation}
\left\{ \begin{array}{l}
Th_1 = \frac{C_h-\beta \delta_{\lambda_0}(Th_1)}{R_h+C_h} \\
Th_2 = \frac{(R_h-R_l)+C_l-\beta \delta_{\lambda_1}(Th_2)}{R_l+C_l}
\end{array}. \right.
\end{equation}

 A simple threshold structure on each edge is clear from Figure 1. Next we will derive the structure of the optimal policy in the whole belief space.

\begin{theorem}
$\Phi_{a}$ is a simple connected region extended from $d_a$ in the belief space $([0,1],[0,1])$, where
\begin{displaymath}
d_a = \left\{ \begin{array}{ll}
 (0,0) & a=B_r \\
 (0,1) & a=B_2 \\
 (1,0) & a=B_1 \\
 (1,1) & a=B_b
\end{array}, \right.
\end{displaymath}
and $\forall (p_1,p_2) \in \Phi_{B_1}$, $p_1 \ge p_2$; $\forall (p_1,p_2) \in \Phi_{B_2}$, $p_1 \le p_2$.
\end{theorem}
\begin{proof}
At the beginning of this section we already show that $d_a \in \Phi_a$, and from Theorem 1, $\Phi_a$ has at least one connected region extended from $d_a$. Therefore next we need to show that each $\Phi_a$  has only one connected region.

For $\Phi_{B_r}$, let $\Phi_{B_r}'$ denote the connected region extended from $(0,0)$, then we will show that there exists no other connected region $\Phi_{B_r}''$. Since $\Phi_{B_r}'$ is symmetric, let $([0,th_1],[0,th_1])$ be the minimum rectangle to include $\Phi_{B_r}'$(no other connected region in this rectangle)(Figure 2(a)). Suppose there is another connected region  $\Phi_{B_r}''$ in the area $([0,th_1],[th_1,1])$ or $([th_1,1],[0,th_1])$. Take the former for example, then we have $\forall (x,y) \in \Phi_{B_r}''$, line $p_1=x$ will pass across both $\Phi_{B_r}'$ and $\Phi_{B_r}''$, thus at least two separate parts of $\Phi_{B_r}$ exist on line $p_1=x$, which contradicts the result in theorem 1. Therefore, no connected region $\Phi_{B_r}''$ exists in area $([0,th_1],[th_1,1])$ or $([th_1,1],[0,th_1])$.

 Suppose another connected region $\Phi_{B_r}''$ exists in $([th_1,1],[th_1,1])$. Let $V_{a}^{p_2=y}(p_1)$ denote $V_{a}(p_1,p_2)$ when $p_2$ is a fixed value $y$. From lemma 1, the slope of line $V_{a}^{p_2=y}(p_1)$ is given by
\begin{equation}
\label{eqn:partialV}
\left\{ \begin{array}{l}
\frac{\partial V^{p_2=y}_{B_r}(p_1)}{\partial p_1} = \beta \frac{\partial V(T(p_1),T(y))}{\partial p_1} \\
\frac{\partial V^{p_2=y}_{B_2}(p_1)}{\partial p_1} = \beta \frac{\partial [(1-y)V(T(p_1),\lambda_0)+yV(T(p_1),\lambda_1)])}{\partial p_1}
\end{array}. \right.
\end{equation}
From (\ref{eqn:partialV}) we have $\frac{\partial V^{p_2=y}_{B_r}(p_1)}{\partial p_1} < \frac{\partial V^{p_2=y}_{B_2}(p_1)}{\partial p_1}$ and from the structure of optimal policy on the boundary of the belief space, we have $\forall (x,y) \in \Phi_{B_r}''$, we have $V_{B_2}(0,y) > V_{B_r}(0,y)$ (Figure 2 (b)).
\psset{unit=80pt}
\begin{figure}[htbp]
\centering
\includegraphics[width=0.5\textwidth,natwidth=914,natheight=449]{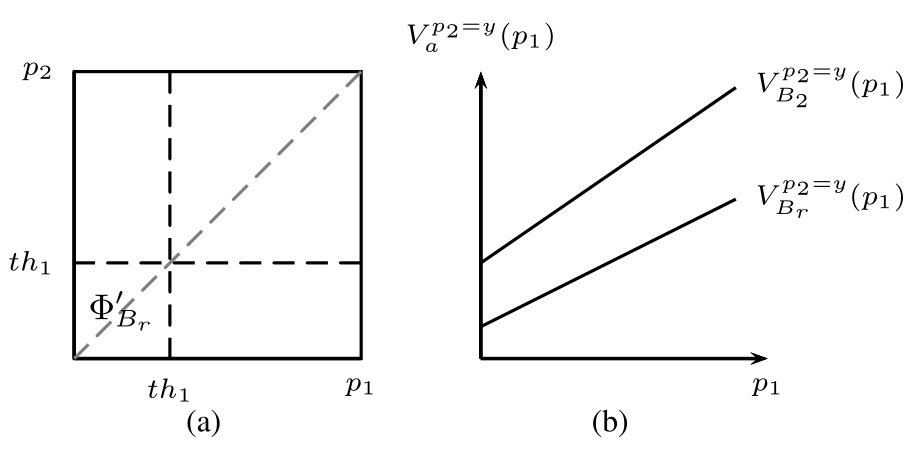}
\caption{(a)Belief space region segmentation (b)$V^{p2=y}_{a}(p_1)$} 
\end{figure}

It is clear from Fig.2(b) that there exists no $p_1=x$ such that $V^{p_2=y}_{B_2}(x)<V^{p_2=y}_{B_r}(x)$. Therefore $(x,y) \notin \Phi_{B_r}$, which contradicts our assumption that $(x,y) \in \Phi_{B_r}'' \subset \Phi_{B_r}$. From this we show there is no other connected region $\Phi_{B_r}''$ in the area $([th_1,1],[th_1,1])$. In other words, $\Phi_{B_r}'$ is the only connected region of $\Phi_{B_r}$.

We can prove that $\Phi_{B_1}, \Phi_{B_2}$ or $\Phi_{B_b}$ has only one connected region in a similar manner and the detail is omitted due to space limit.

Next we prove $\forall (p_1,p_2) \in \Phi_{B_2}$, $p_2 \ge p_1$. Obviously $\Phi_{B_2}$ has a connected region extended from $(0,1)$. If $\exists (x,y) \in \Phi_{B_2}$ and $x>y$, we can find a mirror point of $(x,y)$ with respect to line $p_1=p_2$ according to the convexity of $\Phi_{B_2}$. Then both points $(x,y)$ and $(y,x)$ belong to $\Phi_{B_2}$, which contradicts theorem 2. Hence $\forall (p_1,p_2) \in \Phi_{B_2}$, we have $p_1 \le p_2$. Similarly, $\forall (p_1,p_2) \in \Phi_{B_1}$, $p_1 \ge p_2$.

\end{proof}

When we prove the extended region of $\Phi_{B_b}$ in theorem 3 (refer to [18] for detail), two types of structures are found on the line $p_1=p_2$. (1) one threshold structure: $\exists 0<\rho_1<1$, such that $\forall y \in [0,\rho_1]$, $(y,y) \in \Phi_{B_r}$, and $\forall y \in [\rho_1,1]$, $(y,y) \in \Phi_{B_b}$. (2) two threshold structure: $\exists 0<\rho_1<\rho_2<1$, such that $\forall y \in [0,\rho_1]$, $(y,y) \in \Phi_{B_r}$; $\forall y \in [\rho_1,\rho_2]$, $(y,y) \in \Phi_{B_1}(\Phi_{B_2})$; and $\forall y \in [\rho_2,1]$,$(y,y) \in \Phi_{B_b}$. From  theorem 3, the structure of the optimal policy is illustrated in Figure 3.
\begin{figure}[!t]
\centering
\subfloat[1-threshold Structure]{
\begin{minipage}[t]{0.24\textwidth}
 \includegraphics[width=\textwidth,natwidth=290,natheight=270]{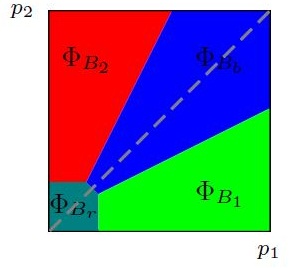}
\end{minipage}
}
\subfloat[2-threshold Structure]{
\begin{minipage}[t]{0.24\textwidth}
 \includegraphics[width=\textwidth,natwidth=290,natheight=270]{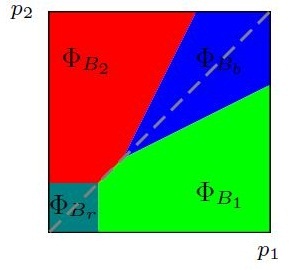}
\end{minipage}
}
\caption{The structure of optimal policy $\pi^{*}$} 
\end{figure}

For the one threshold structure, $\rho_1$ can be obtained by solving $V_{B_r}(\rho_1,\rho_1)=V_{B_b}(\rho_1,\rho_1)$. For the two threshold structure, $\rho_1$ and $\rho_2$ can be obtained by solving $V_{B_r}(\rho_1,\rho_1)=V_{B_1}(\rho_1,\rho_1)$ and $V_{B_1}(\rho_2,\rho_2)=V_{B_b}(\rho_2,\rho_2)$. Therefore for this power allocation problem, we are able to give the basic structure of the optimal policy and derive the thresholds on four edges and on the line $p_1=p_2$. However, so far we are unable to derive a closed form expression for the boundary of each $\Phi_a$.  In the next section, we will use simulation based on linear programming to construct the optimal policy and verify its features.

\section{Simulation Based on Linear Programming}
Linear programming is one of the approaches to solve the Bellman equation in (4). Based on [13], we model our problem as the following linear program:
\begin{equation}
\begin{array}{rll}
\min \sum_{\mathbf{p} \in \mathbb{X}}V(\mathbf{p}), & & \\
s.t. \ \ g_a(\mathbf{p}) &+& \beta \sum_{\mathbf{y} \in \mathbb{X}}f_a(\mathbf{p},\mathbf{y})V(\mathbf{y}) \le V(\mathbf{p}), \\
& &\forall \mathbf{p} \in \mathbb{X}, \forall a \in \mathbb{A}_{\mathbf{p}}
\end{array}
\end{equation}
where $\mathbb{X}$ denotes the belief space, $\mathbb{A}_{\mathbf{p}}$ is the set of available actions for state $\mathbf{p}$. The state transition probability $f_a(\mathbf{p},\mathbf{y})$ is the probability that the next state will be $\mathbf{y}$ when the current state is $\mathbf{p}$ and the current action is $a \in \mathbb{A}_{\mathbf{p}}$. The optimal policy is given by
\begin{equation}
\pi(\mathbf{p}) = \arg \max_{a \in \mathbb{A}_{\mathbf{p}}} (g_a(\mathbf{p}) + \beta \sum_{\mathbf{y} \in \mathbb{X}}f_a(\mathbf{p},\mathbf{y})V(\mathbf{p}) ).
\end{equation}

We used the LOQO solver on NEOS Server [14] with AMPL input [15] to obtain the solution of equation (37). Then we used MATLAB to construct the policy according to equation (38).

\begin{figure}[!t]
\centering
\subfloat[1-threshold form]{
\begin{minipage}[t]{0.24\textwidth}
 \includegraphics[width=\textwidth,natwidth=351,natheight=381]{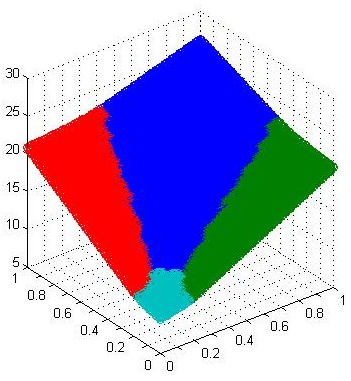} \\
 \includegraphics[width=\textwidth,natwidth=393,natheight=385]{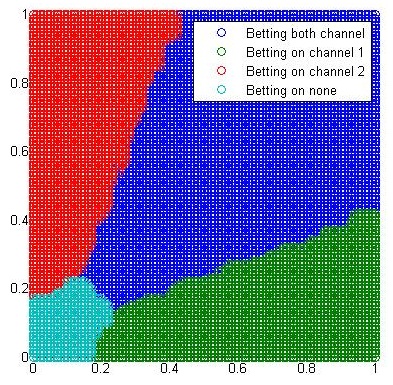}%
\end{minipage}
}
\subfloat[2-threshold form]{
\begin{minipage}[t]{0.24\textwidth}
 \includegraphics[width=\textwidth,natwidth=351,natheight=381]{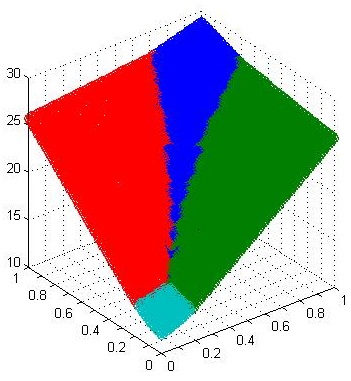}\\
 \includegraphics[width=\textwidth,natwidth=390,natheight=380]{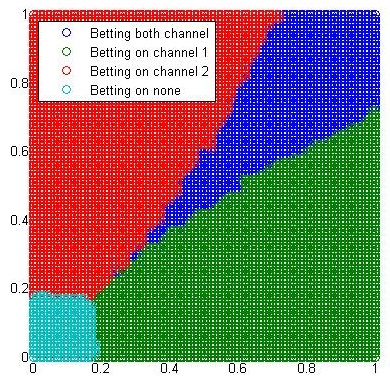}%
\end{minipage}
}
\caption{Value function and Structure of optimal policy}
\end{figure}

Figure 4 shows the AMPL solution of the value function and the corresponding optimal policy. In Fig 4(a), we use the following set of parameters: $\lambda_0=0.1, \lambda_1=0.9, \beta=0.9, R_h/R_l=3/2=1.5, C_h/C_l=1.2/0.8=1.5$ and the ``1-threshold structure'' of the optimal policy is observed; In Fig 4(b), we use the same set of parameters as in (a) except $R_h/R_l=3.7/2=1.85$ and the ``2-threshold structure'' of the optimal policy is observed. The optimal policy in Figure 4 clearly shows the properties we gave in Section 3.

For our power allocation problem, it is interesting to investigate the effect of parameters (such as $\lambda_0,\lambda_1,R_h,R_l,C_h$ and $C_l$) on the structure of optimal policy. For this purpose, we conducted simulation experiments with varying parameters. First, we increase $\lambda_0$ from 0.1 to 0.8 while keeping the rest of parameters the same as in experiment in Fig 4(a). Let $|\Phi_a|$ denote the area of $\Phi_a$ in total belief space and we normalize all $\Phi_a$ with total belief space as $|\Phi_a|/|\mathbb{X}|$. Fig 5(a) shows how the normalized $\Phi_a$ changes with different $\lambda_0$. We can observe in Fig 5(a) that initially $\Phi_{B_b}$ has the biggest area with $\lambda_0=0.1$. When gradually increasing $\lambda_0$, $\Phi_{B_b}$ becomes smaller, whilst $\Phi_{B_1}(\Phi_{B_2})$ and $\Phi_{B_r}$ become bigger. When $\lambda_0 \ge 0.5$, $\Phi_{B_1}(\Phi_{B_2})$ occupies the major part of the belief space, meaning that $\lambda_0$ is big enough and it is more optimal to ``gamble'' on one channel. Similarly, Fig 5(b) shows the results when we decrease $\lambda_1$ from 0.9 to 0.2. $|\Phi_{B}|$ changes in a similar manner as in Figure 5(a). We can see that only when $\lambda_1$ is as small as 0.3, $|\Phi_{B_1}|(|\Phi_{B_2}|)$ is bigger than $|\Phi_{B_b}|$. Interestingly, $|\Phi_{B_r}|$ is always the smallest in both experiments, which means the system likes ``gambling'' instead of ``being conservative''.

\begin{figure}[!t]
\centering
\subfloat[Normalized $\Phi_a$ with increasing $\lambda_0$]{
\begin{minipage}[t]{0.24\textwidth}
 \includegraphics[width=\textwidth,natwidth=1060,natheight=822]{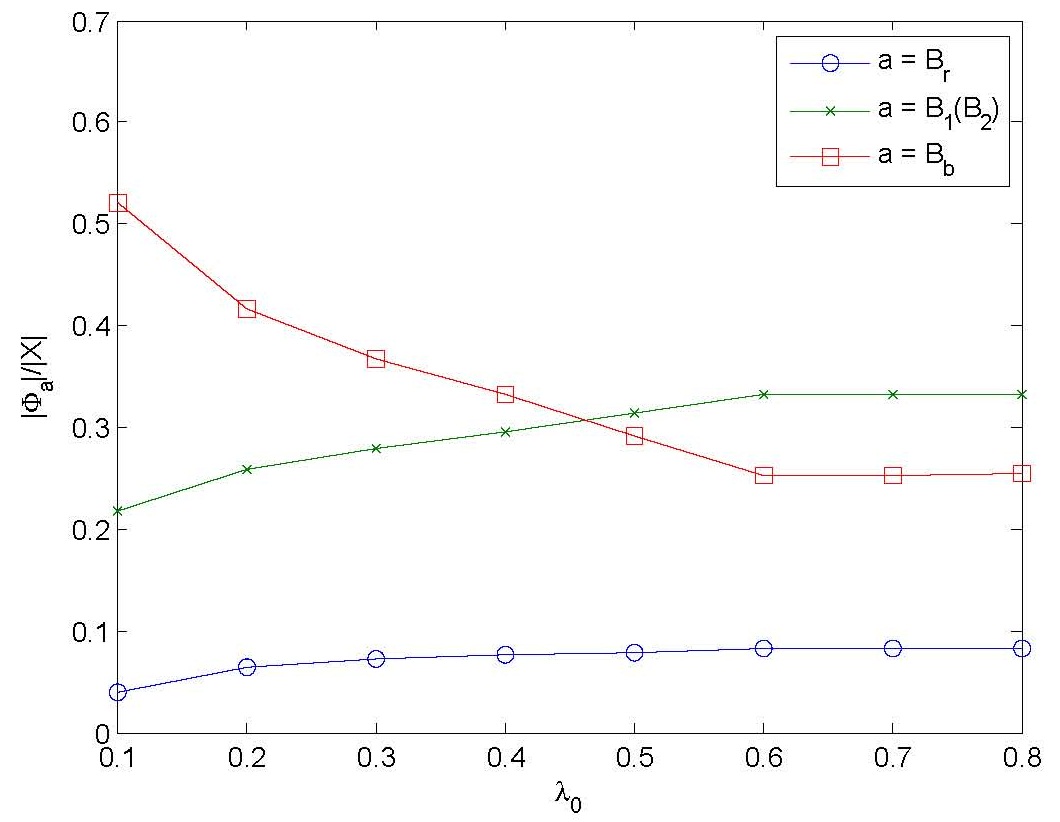}
\end{minipage}
}
\subfloat[Normalized $\Phi_a$ with decreasing $\lambda_1$]{
\begin{minipage}[t]{0.24\textwidth}
 \includegraphics[width=\textwidth,natwidth=1042,natheight=822]{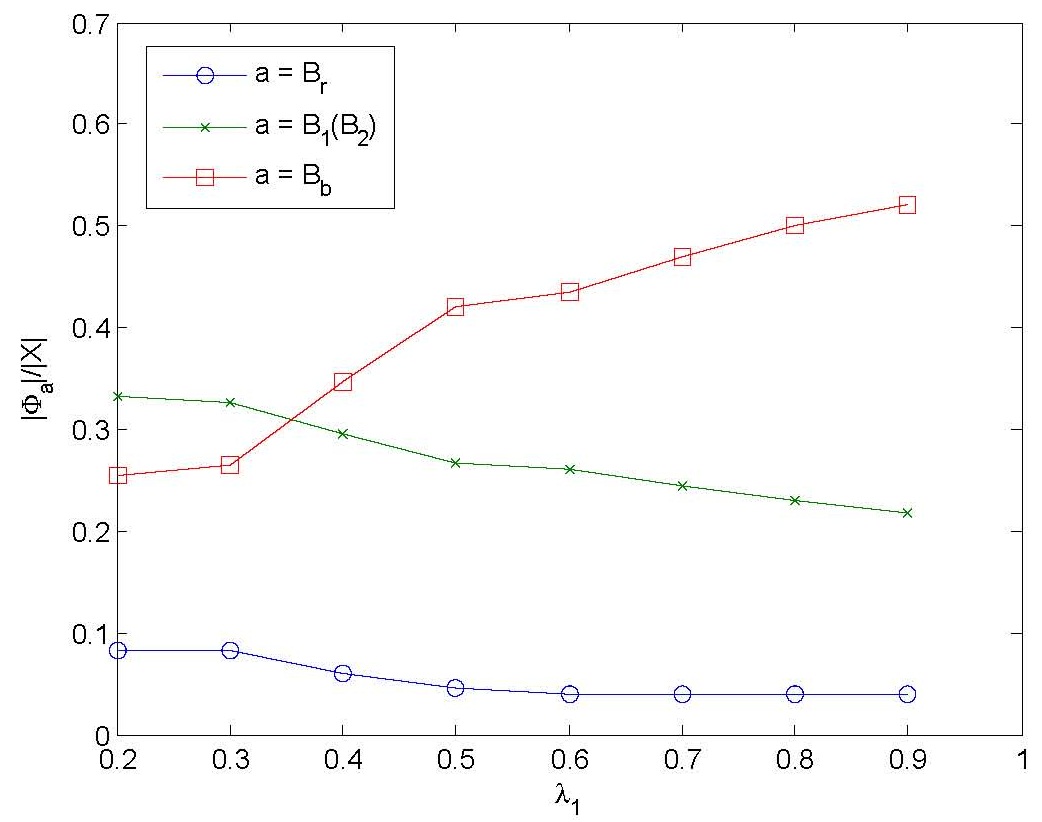}
\end{minipage}
}
\caption{Normalized $\Phi_a$($R_h/R_l=3/2, C_h/C_l=1.2/0.8$)}
\end{figure}

 In Figure 4 we already observed that different $R_h/R_l$ ratio leads to different structure of optimal policy (1-threshold or 2-threshold structure). Therefore we believe the optimal policy is closely related to the immediate reward and loss of the four actions. And we believe the ratio of $R_h/R_l$ and $C_h/C_l$ has more effect on the structure of optimal policy than their real value. Therefore, in the next experiment we increase the ratio of $R_h/R_l$ with different $C_h/C_l$.

\begin{figure}[!t]
\centering
\subfloat[$\Phi_{B_r}$ with increasing $R_h/R_l$]{
\begin{minipage}[t]{0.24\textwidth}
 \includegraphics[width=\textwidth,natwidth=1062,natheight=822]{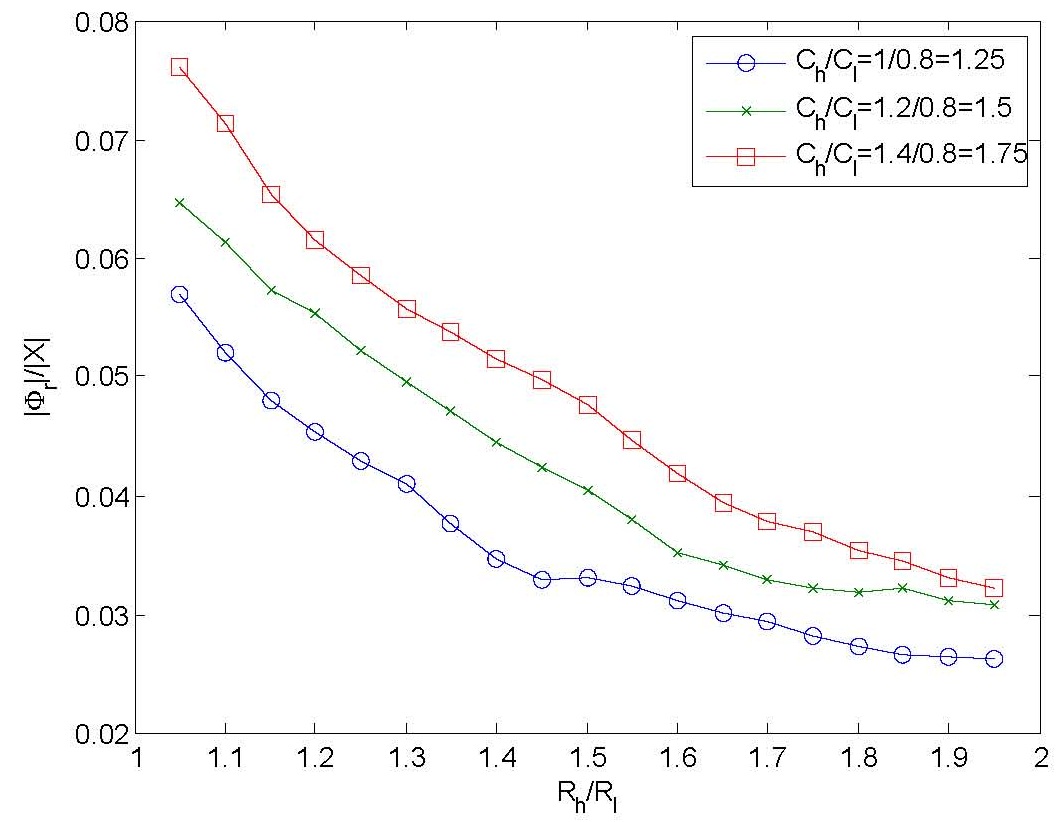}
\end{minipage}
}
\subfloat[$\Phi_{B_r}$ with increasing $C_h/C_l$]{
\begin{minipage}[t]{0.24\textwidth}
 \includegraphics[width=\textwidth,natwidth=1062,natheight=822]{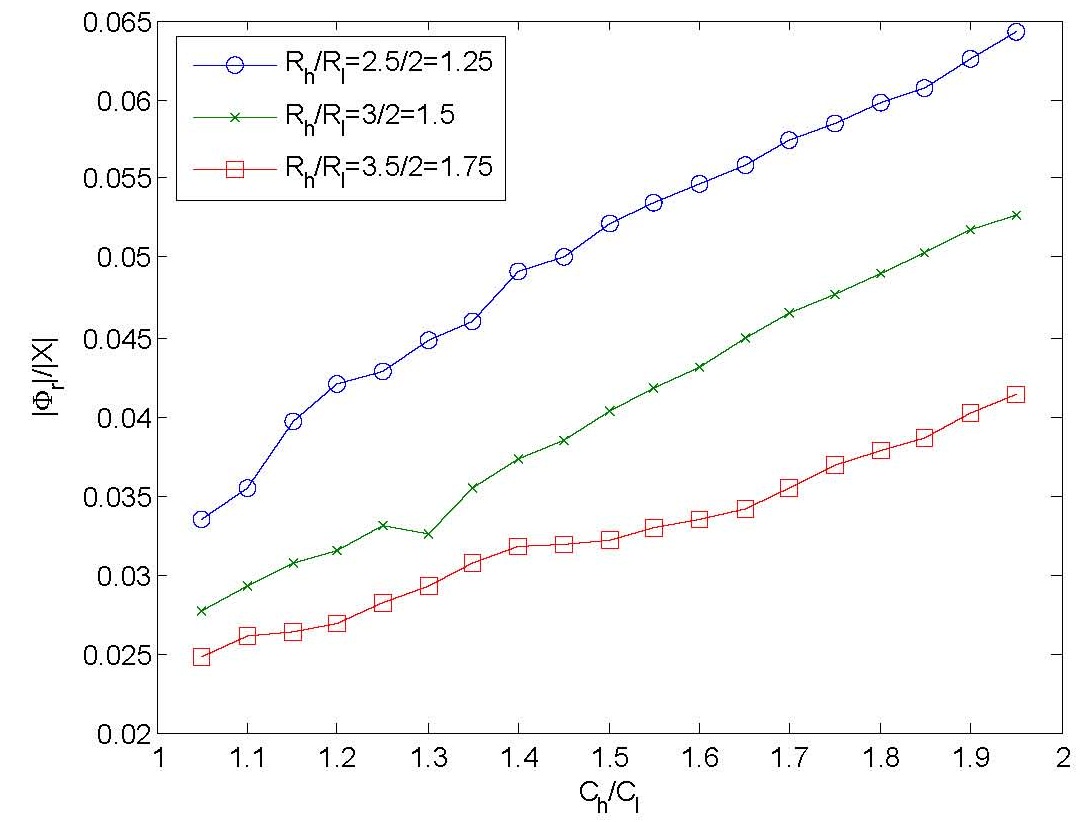}
\end{minipage}
}
\\
\subfloat[$\Phi_{B_1}$ with increasing $R_h/R_l$]{
\begin{minipage}[t]{0.24\textwidth}
 \includegraphics[width=\textwidth,natwidth=1062,natheight=822]{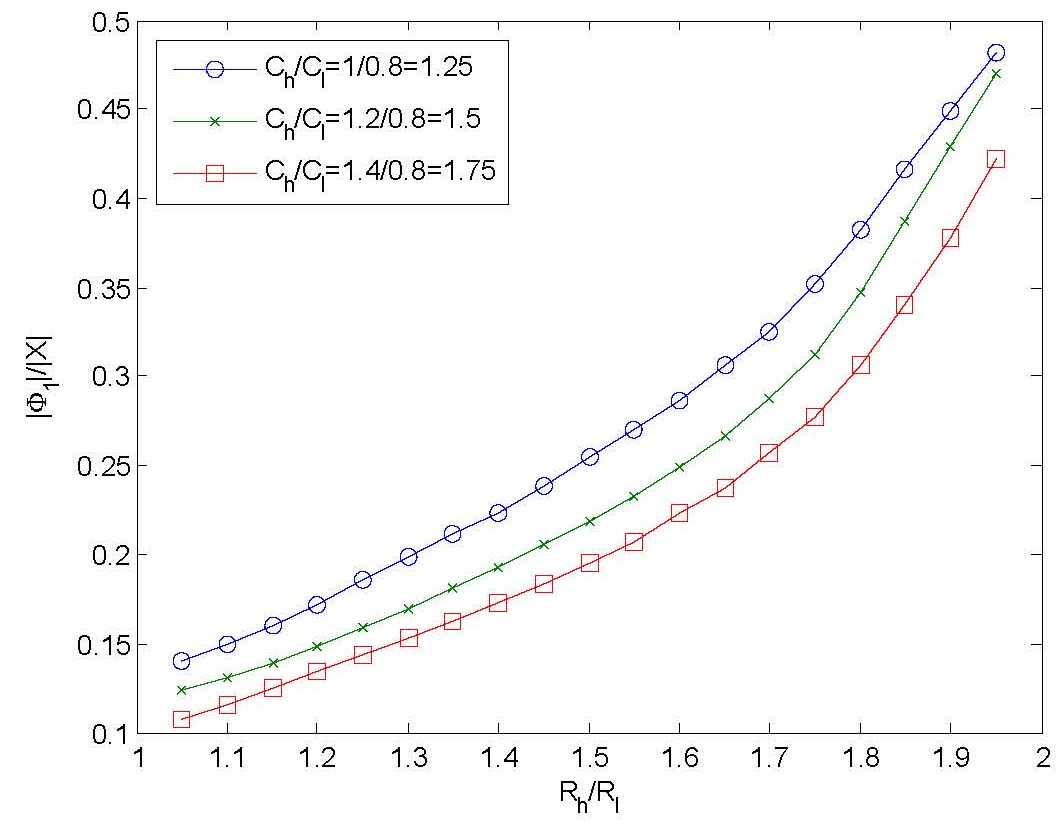}
\end{minipage}
}
\subfloat[$\Phi_{B_1}$ with increasing $C_h/C_l$]{
\begin{minipage}[t]{0.24\textwidth}
 \includegraphics[width=\textwidth,natwidth=1062,natheight=822]{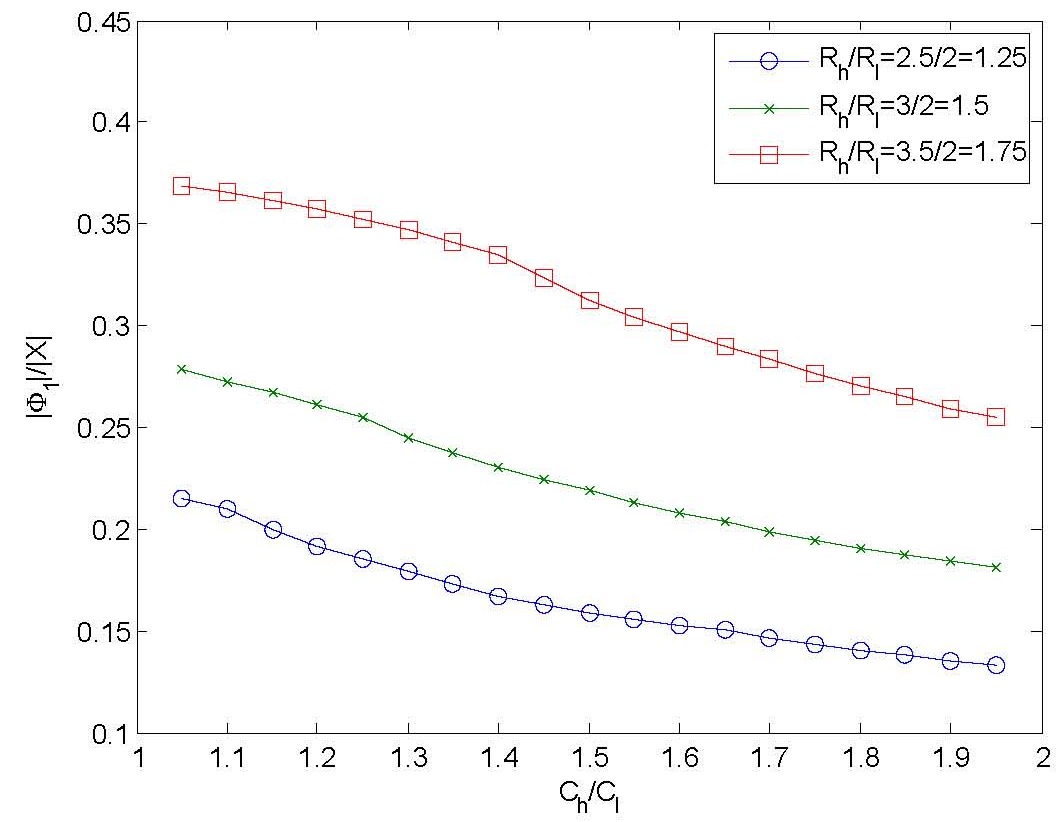}
\end{minipage}
}
\\
\subfloat[$\Phi_{B_b}$ with increasing $R_h/R_l$]{
\begin{minipage}[t]{0.24\textwidth}
 \includegraphics[width=\textwidth,natwidth=1062,natheight=822]{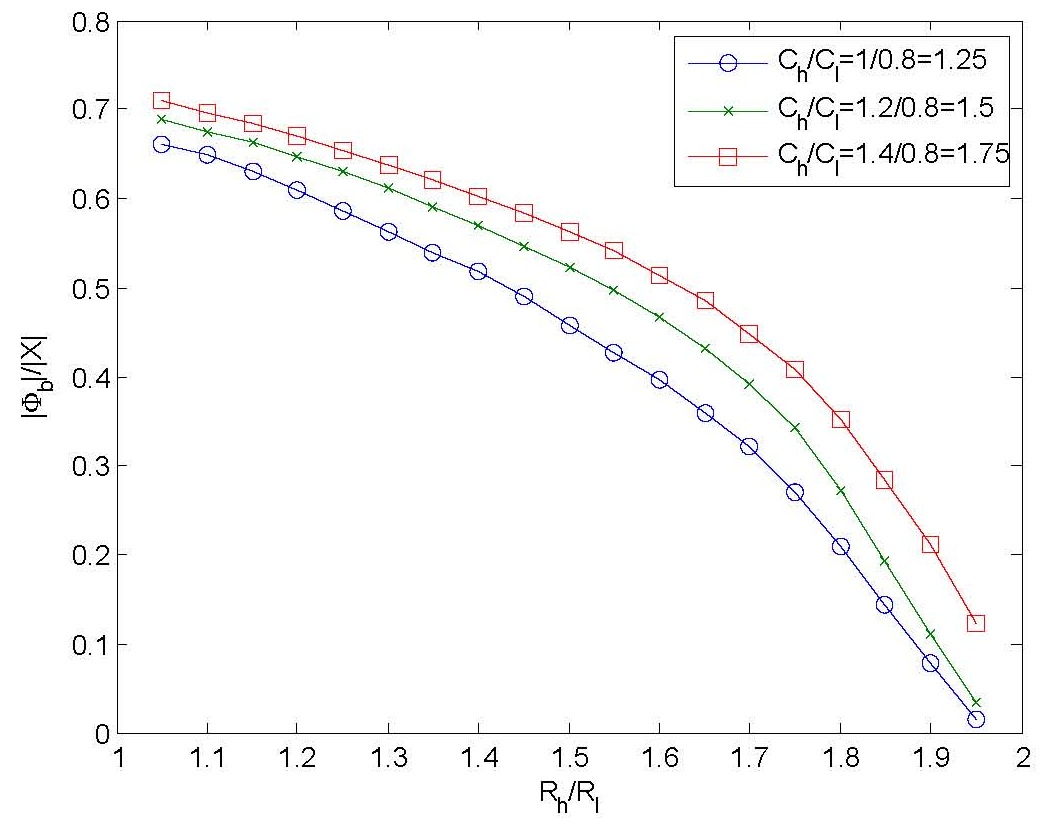}
\end{minipage}
}
\subfloat[$\Phi_{B_b}$ with increasing $C_h/C_l$]{
\begin{minipage}[t]{0.24\textwidth}
 \includegraphics[width=\textwidth,natwidth=1062,natheight=822]{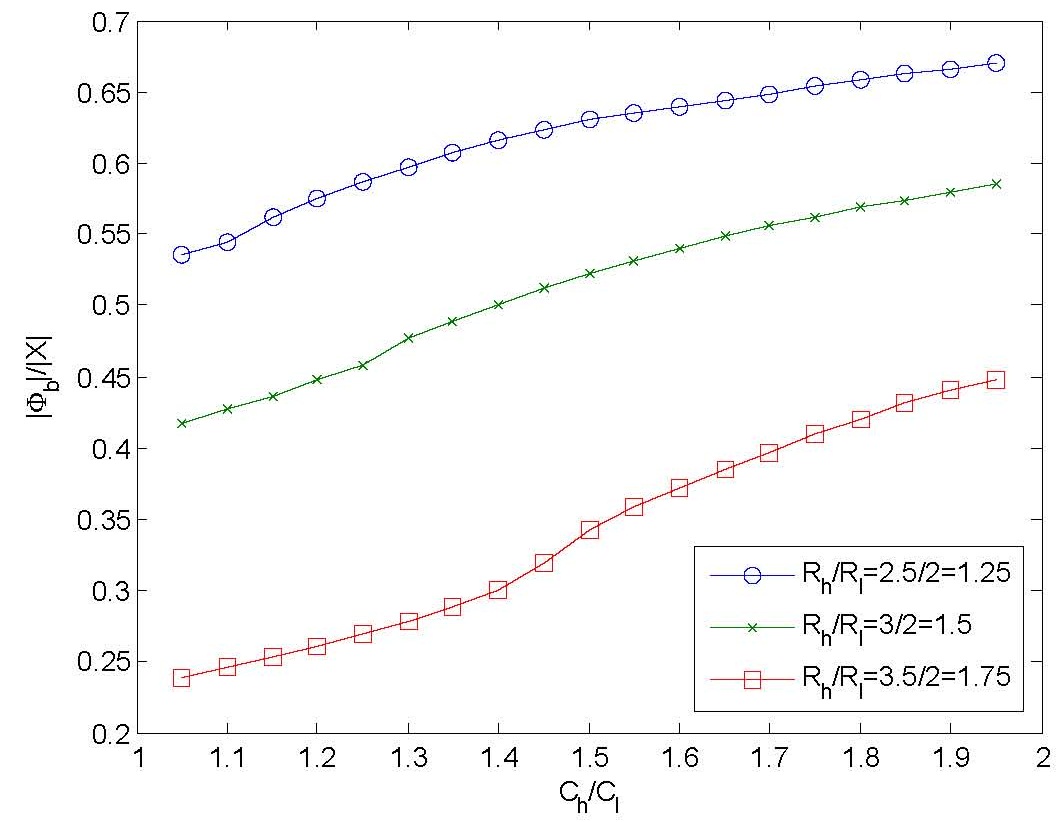}
\end{minipage}
}
\caption{Normalized $\Phi_{a}$ ($\lambda_0=0.1,\lambda_1=0.9$)}
\end{figure}

Figure 6(a)(c)(e) show the normalized $\Phi_{B_r},\Phi_{B_1},\Phi_{B_b}$ with increasing $R_h/R_l$ from 1.05 to 1.95. We can see that when $R_h/R_l$ increases, $|\Phi_{B_r}|$ and $|\Phi_{B_b}|$ become smaller while $|\Phi_{B_1}|$ grows bigger, meaning that the immediate reward of using one channel($R_h$) is big enough to justify ``gambling''. Similarly Figure 6(b)(d)(f) show the normalized $\Phi_{B_r},\Phi_{B_1},\Phi_{B_b}$ with increasing $C_h/C_l$ from 1.05 to 1.95. In contrast to Figure 6(a)(c)(e), when $C_h/C_l$ increases, $|\Phi_{B_r}|$ and $|\Phi_{B_b}|$ grows bigger while $|\Phi_{B_1}|$ becomes smaller, meaning that the the immediate loss of using a channel($C_h$) is big enough and the system decides to ``play safe''.

From above observation we understand that ``1-threshold structure'' may occur with small $R_h/R_l$, big $C_h/C_l$ and $\lambda_1-\lambda_0$; ``2-threshold structure'' may occur with big $R_h/R_l$, small $C_h/C_l$ and $\lambda_1-\lambda_0$. Figure 7 verifies our speculation. We can see that in all experiments with a wide range of parameters, no other policy structure than 1-threshold and 2-threshold structure is observed. So we can conclude that with the help of linear-programming simulation, once the parameters ($\lambda_0,\lambda_1,R_h,R_l,C_h,C_l,\beta$) are known, the structure of optimal policy can be derived like in Figure 7(a)(b).

\begin{figure}[!t]
\centering
\subfloat[1-threshold structure]{
\begin{minipage}[t]{0.24\textwidth}
 \includegraphics[width=\textwidth,natwidth=388,natheight=382]{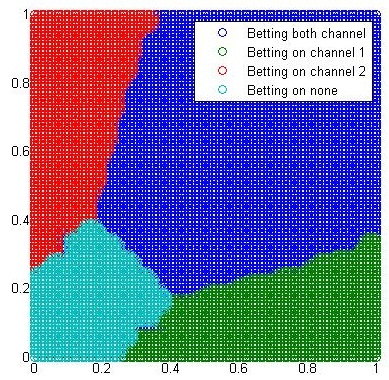}
\end{minipage}
}
\subfloat[2-threshold structure]{
\begin{minipage}[t]{0.24\textwidth}
 \includegraphics[width=\textwidth,natwidth=391,natheight=380]{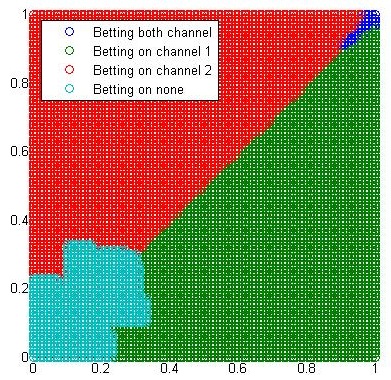}
\end{minipage}
}
\caption{(a)$R_h/R_l=1.25,C_h/C_l=1.95,\lambda_0=0.1,\lambda_1=0.9$; (b)$R_h/R_l=1.95,C_h/C_l=1.5,\lambda_0=0.4,\lambda_1=0.6$}
\end{figure}

\section{Conclusion}
In this paper we have derived the structure of optimal policy for our power allocation problem by theoretical analysis and simulation. We have given the structure of optimal policy on total belief space and proved that the optimal policy for this problem has a 1 or 2 threshold structure. With the help of linear programming, we can derive the optimal policy with key parameters. Further, we would like to find a closed form expression for the boundary of action region. Also, we would like to investigate the case of non-identical channels like [16], or derive useful results for more than 2 channels.

\section*{Acknowledgment}
 This work is partially supported by Natural Science Foundation of China under grant 61071081 and 60932003. This research was also sponsored in part by the U.S. Army Research Laboratory under the Network Science Collaborative Technology Alliance, Agreement Number W911NF-09-2-0053, and by the Okawa Foundation, under an Award to support research on ``Network Protocols that Learn".





\begin{thebibliography}{1}
\bibitem{Adaptive power control definition}
T. Yoo and A. Goldsmith, "Capacity and power allocation for fading mimo channels with channel estimation error,"
\emph{IEEE Transactions on Information Theory}, vol. 52, pp. 2203-2214, May 2006.

\bibitem{Adaptive power control definition}
W. Yu, W. Rhee, S. Boyd, and J. M. Cioffi, "Iterative water-filling for gaussian vector multiple-access channels,"
\emph{IEEE Transactions on Information Theory}, vol. 50, no. 1, pp. 145-152, 2009.

\bibitem{recent work}
I. Zaidi and V. Krishnamurthy, "Stochastic adaptive multilevel waterfilling in mimo-ofdm wlans," in
\emph{39th Asilomar Conference on Signals, Systems and Computers}, 2005.

\bibitem{recent work}
X. Wang, D. Wang, H. Zhuang, and S. D. Morgera, "Energy-efficient resource allocation in wireless sensor networkds over fading tdma,"
\emph{IEEE Journal on Selected Areas in Communications (JSAC)}, vol. 28, no. 7, pp. 1063-1072, 2010.

\bibitem{recent work}
Y. Gai and B. Krishnamachari, "Online learning algorithms for stochastic water-filling," in
\emph{Information Theory and Application Workshop (ITA 2012)}, 2012.

\bibitem{betting on}
A. Laourine and L. Tong, "Betting on gilbert-elliot channels,"
\emph{IEEE Transactions on Wireless communications}, vol. 9, pp. 723-733, February 2010.

\bibitem{Gillbert-Elliott}
E. N. Gilbert, "Capacity of a burst-noise channel,"
\emph{Bell Syst. Tech. J.}, vol. 53, pp. 1253-1265, Sep 1960.

\bibitem{multi-armed}
Q. Zhao, B. Krishnamachari, and K. Liu, "On myopic sensing for multi-channel opportunistic access: Structure, optimality, and performance,"
\emph{IEEE Transactions on Wireless Commuications}, vol. 7, no. 12, pp. 5431-5540, 2008.

\bibitem{multi-armed}
S. H. A. Ahmad, M. Liu, T. Javidi, Q. Zhao, and B. Krishnamachari, "Optimality of myopic sensing in multi-channel opportunistic access,"
\emph{IEEE Transactions on Information Theory}, vol. 55, no. 9, pp. 4040-4050, 2009.

\bibitem{unknown probabilities}
Y. Wu and B. Krishnamachari, "Online learning to optimize transmission over unknown gilbert-elliott channel," in
\emph{WiOpt}, 2012.

\bibitem{sufficient statistic}
R. D. Smallwood and E. J. Sondik, "The optimal control of partially observable markov precesses over a finite horizon,"
\emph{Operations Research}, vol. 21, pp. 1071-1088, September-October 1973.

\bibitem{optimal policy}
S. M. Ross, \emph{Applied Probability Models with Optimization Applications.} San Francisco: Holden-Day, 1970.

\bibitem{AMPL}
D. P. D. Farias and B. V. Roy, "The linear programming approach to approximate dynamic programming,"
\emph{Operations Research}, vol. 51, pp. 850-865, November-December 2002.

\bibitem{NEOS}
"Neos server for optimization." http://neos.mcs.anl.gov/neos/.

\bibitem{ampl}
R. Fourer, D. M. Gay, and B. W. Kernighan,
\emph{AMPL: A Modeling Language for Mathematical Programming.} Brooks/Cole Publishing Company, 2002.

\bibitem{ampl}
N. Nayyar, Y. Gai, and B. Krishnamachari, "On a restless multi-armed bandit problem with non-identical arms," in
\emph{Allerton}, 2011.

\bibitem{junhuaiccfull}
J.~Tang, P.~Mansourifard, and B.~Krishnamachari, ``Power allocation over two identical gilbert-elliott channels,'' in {\em
  http://arxiv.org/abs/1203.6630}, 2012.

\bibitem{jiangweiiccfull}
W.~Jiang, J.~Tang, and B.~Krishnamachari, ``Optimal Power allocation Policy over two identical gilbert-elliott channels,'' in {\em
  http://arxiv.org/}, 2012.


\end{thebibliography}
%

\end{document}